\PassOptionsToPackage{table}{xcolor}
\documentclass[sigconf,screen]{acmart}
\renewcommand\footnotetextcopyrightpermission[1]{} 
\acmYear{2026}
\acmConference[Preprint]{Preprint}{May}{2026}

\usepackage[ruled,linesnumbered]{algorithm2e}
\usepackage{booktabs}
\usepackage{enumitem}
\usepackage{amsthm}

\usepackage{amssymb}
\usepackage{subcaption}

\newtheorem{theorem}{Theorem}
\newtheorem{proposition}[theorem]{Proposition}

\theoremstyle{definition}

\newtheorem{example}{Example}
\theoremstyle{remark}
\newtheorem{remark}[theorem]{Remark}

\begin{document}

\title{FLARE: One-Shot PE-Level Fault Localization in Systolic Arrays via Algebraic Test Vectors}

\author{Logashree Venkatasubramanian}
\affiliation{%
  \institution{Georgia Institute of Technology}
  \department{School of Electrical and Computer Engineering}
  \city{Atlanta}
  \state{Georgia}
  \country{USA}
}
\email{lvenkata7@gatech.edu}

\author{Zishen Wan}
\affiliation{%
  \institution{Georgia Institute of Technology}
  \department{School of Electrical and Computer Engineering}
  \city{Atlanta}
  \state{Georgia}
  \country{USA}
}
\email{zishenwan@gatech.edu}

\author{Viveck Cadambe}
\affiliation{%
  \institution{Georgia Institute of Technology}
  \department{School of Electrical and Computer Engineering}
  \city{Atlanta}
  \state{Georgia}
  \country{USA}
}
\email{viveck@gatech.edu}

\begin{abstract}
Systolic arrays are the dominant compute fabric for neural network inference. Prior work has addressed \emph{column-level} fault detection efficiently with uniform test patterns, but \emph{row-level} (PE-level) fault localization within a faulty column remains open without resorting to hardware redundancy. The fundamental obstacle is that uniform test inputs destroy per-row signatures: any test that activates every row equally cannot distinguish which row is the source of an observed deviation.

In this paper, we propose a lightweight, purely algorithmic remedy based on \emph{coprime test vectors}. By assigning pairwise coprime integers as test-input entries, a permanent weight-register fault produces a deviation whose divisibility signature uniquely identifies the faulty row. Under a general bounded error model, a single test pass localizes the faulty row with high probability. This error model covers a broader class of faults than what prior dataflow-aware testing work has primarily emphasized. When one round is insufficient, a second pass using a ratio computation achieves exact localization; for the special case of single-bit errors, odd coprime entries guarantee exact localization in one round.

For INT16 arithmetic, a single test pass covers array sizes up to $256{\times}256$ with localization probability above $0.98$, at a test cost under $1\%$ of one inference GEMM tile.
\end{abstract}




\maketitle

\section{Introduction}
\label{sec:intro}

Given the rapid growth of artificial intelligence workloads, particularly large language models, there has been a significant push toward designing power, performance, and area (PPA) efficient accelerators tailored to these workloads~\cite{jouppi2017tpu, deng2020model,qin2024mecla,wan2025cogsys}. These accelerators are increasingly deployed in safety critical and time sensitive applications, including autonomous driving systems, space missions, medical robotics, and other embedded platforms where incorrect predictions can have severe consequences. Consequently, ensuring the reliability of AI accelerators has become a critical research challenge. Industry-scale studies from Google, Meta, and Alibaba reveal that silent data corruption (SDC), hardware faults producing incorrect results without error flags, occur at significant rates across major computing fleets, prompting dedicated teams at these companies for fault detection and localization approaches~\cite{hochschild2021cores,dixit2021silent}. {Permanent faults from latent manufacturing defects are of particular concern: such defects can escape post-silicon screening yet cause persistent silent data corruption in deployment~\cite{hochschild2021cores,dixit2021silent,he2023understanding}. Studies show that even a single faulty PE can meaningfully degrade DNN inference accuracy~\cite{zhang2018analyzing}, making PE-level fault localization a critical prerequisite for targeted mitigation.}

Systolic arrays have emerged as a dominant architectural paradigm for accelerating matrix intensive workloads, particularly in machine learning and deep neural networks~\cite{kung1982systolic,raj2025scale}. By organizing computation around a regular grid of Processing Elements (PEs), systolic arrays enable high throughput and energy efficient matrix multiplication, which forms the computational backbone of modern AI models. Industry grade AI accelerators, such as Google's Tensor Processing Units (TPUs), are architected as large scale systolic arrays, highlighting the practical relevance of systolic array centric reliability challenges~\cite{jouppi2017tpu}.

A common approach to fault tolerance relies on computational or hardware redundancy. Classical Algorithm Based Fault Tolerance (ABFT), initially proposed by Huang and Abraham~\cite{huang1984abft}, appends row and column checksums to input matrices and verifies computed outputs against these checksums to detect errors. Notably, ABFT incurs compute and memory overhead for checksum generation and verification, and the recovery step (when present) requires additional passes over the output.  Hardware based approaches such as Double Modular Redundancy (DMR) or Triple Modular Redundancy (TMR) consume 100--200\% additional area~\cite{constantinescu2003trends}, directly undermining the PPA efficiency that makes systolic arrays attractive in the first place.

To preserve PPA while tolerating faults, recent work introduces spare PE rows or columns and reroutes computation around known faulty PEs~\cite{goldstein2021lightweight, bal2023novel, cherezova2025fortalesa, lee2024area, cho2020efficient, zhao2022fsa}. These mitigation strategies, whether rerouting computation or adapting weight mappings to avoid faulty PEs~\cite{ait2025algorithmic}, typically require the precise PE location as a prerequisite. Identifying the faulty PE requires taking the device offline and running diagnostic routines or BIST-based structural tests~\cite{hochschild2021cores, dixit2021silent, lee2023strait}, or relying on dedicated hardware checkers, probes, or scan chains (including more recent augmented-array designs~\cite{liu2025runtime, solangi2021test}) that impose area and power overhead even when no faults are present. Recent algorithmic alternatives~\cite{vacca2023runsafer, peltekis2025periodic} avoid this hardware cost by designing test patterns that expose faults during execution. However, they achieve only column-level fault detection using uniform test vectors, and cannot achieve \emph{PE-level localization} --- identifying which specific row within a faulty column contains the defective PE.

The core challenge is that systolic accumulation inherently destroys row-level identifiability: conventional test patterns collapse all row contributions into a single scalar, making localization impossible.  In this paper, we develop a framework that actively designs input test vectors whose algebraic structure survives this accumulation, so that fault location can be obtained from the structure of the resulting output deviation. Our work addresses a complementary limitation of prior dataflow-aware testing approaches by enabling PE-level localization for weight-register faults without hardware augmentation.
\vspace{2mm}

\noindent\textit{Summary of Contributions.} ~~This work introduces a novel algorithmic framework to achieve \emph{PE-level fault localization} in systolic arrays without hardware redundancy. We focus on permanent weight register (\textsf{Wreg}) faults in weight-stationary systolic arrays, the dominant dataflow in commercial AI accelerators where weights are retained locally within PEs to maximize data reuse and energy efficiency. Our key technical contribution is the development of carefully structured test vectors where fault deviations carry distinctive algebraic signatures that reveal both the location and magnitude of the underlying hardware defect. Our main contributions are as follows.

FLARE assumes at most one faulty weight register per column, with the error magnitude bounded within the representable precision range. Prior works such as RunSAFER~\cite{vacca2023runsafer} and Periodic Online Testing~\cite{peltekis2025periodic} validate only against single-bit stuck-at faults, whereas FLARE covers any error within this range, with single-bit faults as a special case.

\vspace{1mm}
\textbf{Single-round structured localization:} For an $L \times L$ systolic array performing $b$-bit integer computations, we develop a scheme where one test vector obtains the precise fault location with quantified probability, assuming at most one fault per column. Our framework uses pairwise coprime test vector entries to create unique divisibility signatures for each row. For INT16 computations, single-round localization covers typical accelerator array sizes ($L \le 256$) with probability of successful localization above $0.98$. Further, for permanent \emph{single-bit} faults, our approach achieves exact localization in a single round, exploiting the power-of-two structure of these errors. This single-round result is particularly significant for safety-critical deployments where strict time budgets are allocated for fault detection (e.g., ISO 26262 specifies fault tolerant time interval (FTTI) constraints of 100--256\,ms).

\vspace{1mm}
\textbf{Two-round extension for complete coverage:} When single-round localization is insufficient, a second test pass with a simple ratio computation achieves exact localization for all practical array sizes. This adapts the classical principle --- familiar from ABFT and syndrome-based error correction --- that two independent observations suffice to determine both the location and value of a single error.

\vspace{1mm}
\enlargethispage{\baselineskip}
Beyond localization, the framework estimates the magnitude of the hardware defect, enabling selective mitigation strategies where only significant faults affecting accuracy require correction. Thus, FLARE fills the critical gap between column-level detection and the localization requirements of fault-tolerant routing schemes, providing the precise faulty PE identification at run-time without hardware redundancy and minimal test overhead. We validate the analytical bounds through fault injection experiments on a custom cycle-accurate systolic array simulator using weights from \path{Qwen2.5-0.5B}; empirical failure rates are consistent with the theoretical predictions across all tested array dimensions. Cycle overhead measured via SCALE-Sim stays at most 2 cycles per test round, below $0.1\%$ of the cycles required for a single inference GEMM tile at $128{\times}128$ and $256{\times}256$ array sizes. We note that our results apply specifically to weight-register faults in weight-stationary dataflows; other fault types (e.g., input-register or MAC faults) exhibit different error structures and present a natural avenue for extending this framework.

\vspace{-2mm}
\paragraph{Paper outline.}
Section~\ref{sec:related} surveys prior work on ABFT, hardware redundancy, dataflow aware testing, silent data corruption, and DNN resilience under faults. Section~\ref{sec:model} states the system and fault models and presents the formal problem formulation. Section~\ref{sec:method} presents the FLARE one round and two round fault localization algorithms, along with formal theorems that bound the probability of successful localization as a function of array dimension and compute precision. Section~\ref{sec:inst} instantiates FLARE for INT8 and INT16, validates the analytical bounds through fault injection experiments on a custom systolic array simulator, and characterizes the cycle overhead using SCALE-Sim.

\section{Related Work}
\label{sec:related}

\subsection{Fault Detection and Correction}
Algorithm-Based Fault Tolerance (ABFT)~\cite{huang1984abft} augments computation with checksum redundancy to detect and correct errors at the output. Subsequent work reduces overhead and adapts ABFT to modern workloads, including efficient encoding for systolic arrays~\cite{libano2023efficient, braun2014abft, safarpour2021algorithm}, and approximate variants for quantized inference~\cite{xue2023approxabft} and layer-specific schemes for transformers~\cite{liu2024alberta, ma2023error}. Statistical methods have also been used to selectively protect vulnerable layers in LLM inference~\cite{xie2025realm}. These approaches are effective for transient faults, where correcting the output suffices.

ABFT operates at the output level and cannot localize faults to the PE level; hardware approaches achieve PE-level localization but incur area overhead. Our work bridges this gap with an algorithmic solution that delivers PE-level localization without hardware augmentation.

Arithmetic (AN) and residue codes~\cite{massey1972error,iacobovici2015residue} also introduce algebraic redundancy to detect computation errors via consistency checks. These approaches use coprime and modular structure with the goal of invariant checking, whereas our goal is fault localization. While AN and residue codes bear some high level resemblance with our test vector constructions, we see no direct technical connection between the two approaches.

\subsection{Dataflow-Aware Testing and Localization}
\label{sec:colvsrow}

An alternative to output-level detection is to exploit systolic dataflow to design test inputs that expose faults. RunSAFER~\cite{vacca2023runsafer} and Periodic Online Testing~\cite{peltekis2025periodic} inject structured (often uniform) test patterns to detect faults during execution, but achieve only coarse-grained localization at the column level. More broadly, prior work~\cite{tan2023saca, taheri2024saffira} emphasizes runtime detection mechanisms compatible with accelerator datapaths, similarly limited to column or module-level granularity.

This limitation arises from the accumulation inherent in weight-stationary dataflow. Each column computes $y_j = \sum_i w_{ij} x_i$,
so uniform inputs (e.g., all-ones) cause all rows to contribute equally. As a result, row-level information is lost, and the observed deviation indicates only that \emph{some} PE in the column is faulty. Achieving PE-level localization therefore requires test inputs that preserve distinct per-row signatures. This is precisely the gap addressed by our approach.

Offline structural testing, scan-chain-based diagnosis~\cite{lee2023strait, solangi2021test}, and hardware-augmented runtime approaches~\cite{liu2025runtime} can all achieve PE-level localization, but each requires dedicated hardware: for instance, HW-Runtime~\cite{liu2025runtime} augments the array with spare PEs, a control hub, and on-chip comparators, incurring approximately 1.74\% area overhead for a $256\times256$ array. ML-based frameworks such as DiagNNose~\cite{kundu2023diagnnose} offer architecture-specific fault attribution but require fault injection campaigns to generate training data. Across all these approaches, the focus has been limited to single-bit stuck-at faults~\cite{vacca2023runsafer, peltekis2025periodic, liu2025runtime}.

Single-bit faults produce an error that is always a power of two. {While the underlying detection mechanisms of previous works may be useful for multi-bit or arbitrary faults,  formal guarantees for such settings have not been established.} Importantly, real hardware faults may produce arbitrary bounded errors due to multi-bit defects, bridging faults, or parametric variation. In contrast to previous work, our framework explicitly considers such general faults via a general bounded error model ($|e| \le M$), subsuming single-bit faults as a special case. When the error is a single-bit fault, our method provides strictly stronger guarantees (exact localization in one round) as listed in Table~\ref{tab:comparison}. For weight-register faults, FLARE achieves the same localization granularity as HW-Runtime with no hardware overhead and fewer test runs under a more general fault model. In contrast to prior approaches, FLARE also detects the exact deviation of the weight register and can therefore potentially enable selective mitigation strategies (Sec.~\ref{sec:dnn_resilience}). As noted in the introduction, existing methods~\cite{vacca2023runsafer, peltekis2025periodic, liu2025runtime} target broader fault classes beyond weight-register faults, specifically input-register, MAC, and accumulator register faults. The task of advancing the key conceptual insight of FLARE  -- that carefully designing algebraic test vectors tailored to the fault model can provide efficient fault localization -- to models beyond weight-register faults is an open research direction motivated by our paper.

\begin{table}[t!]
\centering\footnotesize
\caption{Comparison of dataflow-aware fault testing approaches for weight-stationary systolic arrays. \vspace{-5pt}}\label{tab:comparison}
\renewcommand*{\arraystretch}{0.8}
\resizebox{\columnwidth}{!}{%
\begin{tabular}{lccc>{\columncolor{yellow!20}}c}
\toprule
 & \shortstack[c]{RunSAFER\\\cite{vacca2023runsafer}}
 & \shortstack[c]{Periodic Online\\Testing~\cite{peltekis2025periodic}}
 & \shortstack[c]{HW-Runtime\\\cite{liu2025runtime}}
 & \shortstack[c]{FLARE\\(this work)} \\
\midrule
\shortstack[l]{Fault model\\(validated)} & \shortstack[c]{single-bit\\stuck-at} & \shortstack[c]{single-bit\\stuck-at} & \shortstack[c]{single-bit\\stuck-at} & \shortstack[c]{bounded error\\(any $|e| \le M$)} \\[2pt]
\shortstack[l]{Test overhead\\(cycles)} & 2 & 4 & 4 $\times$ array-size & 1 or 2 \\[2pt]
Localization & PE col & PE col & PE row \& col & PE row \& col, $|e|$ \\[2pt]
Hw.\ modif. & no & no & yes & no \\
Error est. & no & no & no & yes \\
\bottomrule
\end{tabular}}
\vspace{-5mm}
\end{table}

\subsection{DNN Resilience and Fault Impact}
\label{sec:dnn_resilience}

Large-scale studies of silent data corruption (SDC)~\cite{hochschild2021cores,dixit2021silent} show that existing systems can detect that errors occur but typically cannot localize their origin. In systolic-array-based accelerators, prior work has also analyzed the impact of permanent faults on inference accuracy and proposed mitigation strategies such as retraining or redundancy~\cite{zhang2018analyzing}.

DNN inference exhibits inherent resilience to hardware faults, particularly for low-order bit errors~\cite{li2017understanding,hong2019terminal, hoang2020ft, ozen2020just}. The impact of a fault depends strongly on its magnitude and location: high-order bit errors (e.g., sign or MSB) can significantly degrade accuracy, while low-order perturbations often have negligible effect. Sensitivity also varies across layers and parameters~\cite{wan2023berry, chen2021low, ozen2020boosting}, motivating selective protection strategies that focus on the most critical components.

To quantify fault impact, the Parameter Vulnerability Factor (PVF)~\cite{jiao2024pvf} identifies which model parameters most affect inference accuracy when corrupted. However, PVF does not indicate where the fault originates in hardware. Low-cost runtime monitoring approaches \cite{geissler2023low} can detect such corruptions with minimal overhead, but do not localize the fault source. Our framework complements such analysis by localizing the faulty PE and estimating the error magnitude, enabling targeted mitigation that prioritizes high-impact faults while avoiding unnecessary correction.

\section{System Model and Fault Model}
\label{sec:model}

\subsection{Weight-Stationary Systolic Array}
\label{sec:systolic}

An $L \times K$ systolic array contains $L\cdot K$ processing elements (PEs) arranged in a grid.
PE$(i,j)$, with $i\in[L]$ and $j\in[K]$, holds three components:
\begin{enumerate}[leftmargin=*,nosep]
\item a \emph{weight register} (\textsf{Wreg}), storing a signed integer~$w_{ij}$;
\item an \emph{input register} (\textsf{Ireg}), receiving the signed integer activation~$x_i$ that flows from left to right across row~$i$;
\item a \emph{multiply-accumulate unit} (\textsf{MAC}), computing $w_{ij}\cdot x_i$ and adding to a partial sum flowing top-to-bottom in column~$j$.
\end{enumerate}
Given input $\mathbf{x}=(x_1,\ldots,x_L)^T$, column~$j$ produces
\begin{equation}\label{eq:output}
  y_j = \sum_{i=1}^{L} w_{ij}\,x_i,  \qquad j=1,\ldots,K.
\end{equation}
Both \textsf{Wreg} and \textsf{Ireg} operands are \emph{signed} $b$-bit integers in $\{-2^{b-1},\allowbreak \ldots, \allowbreak2^{b-1}-1\}$ (two's complement); the accumulator is wider (e.g.\ 32\,bits), so~\eqref{eq:output} incurs no overflow.
Throughout, ``INT$b$'' refers to the signed $b$-bit representation.

\subsection{Fault Model}
\label{sec:faultmodel}

We focus on permanent weight-register faults in weight-stationary systolic arrays. In this dataflow, each PE stores a weight value that persists for the duration of a tile computation. A permanent defect in the weight register produces a deviation with a multiplicative structure that our algebraic framework exploits.
Regardless of the physical fault mechanism, a permanent \textsf{Wreg} fault in PE$(i^\ast,j^\ast)$ can be modeled as an additive perturbation to the stored weight:
\begin{equation}\label{eq:deviation}
  \tilde{w}_{i^\ast j^\ast} = w_{i^\ast j^\ast} + e, \qquad
  \Delta_{j^\ast} \triangleq \tilde{y}_{j^\ast} - y_{j^\ast} = e\cdot x_{i^\ast},
\end{equation}
with $\Delta_j=0$ for $j\neq j^\ast$.%
\footnote{A stuck-at fault produces $e\neq 0$ only when the loaded weight disagrees with the stuck bit; otherwise the fault is \emph{masked}. Ensuring activation (i.e., $e\neq 0$) is a standard test-engineering concern addressed by loading complementary weight patterns such as all-zeros and all-ones. Our framework assumes the fault is activated ($e\neq 0$) and focuses on localizing the faulty row from the observed deviation.}
We analyze two error models for $e$.

\begin{figure}[t]
    \centering
    \includegraphics[width=0.7\linewidth]{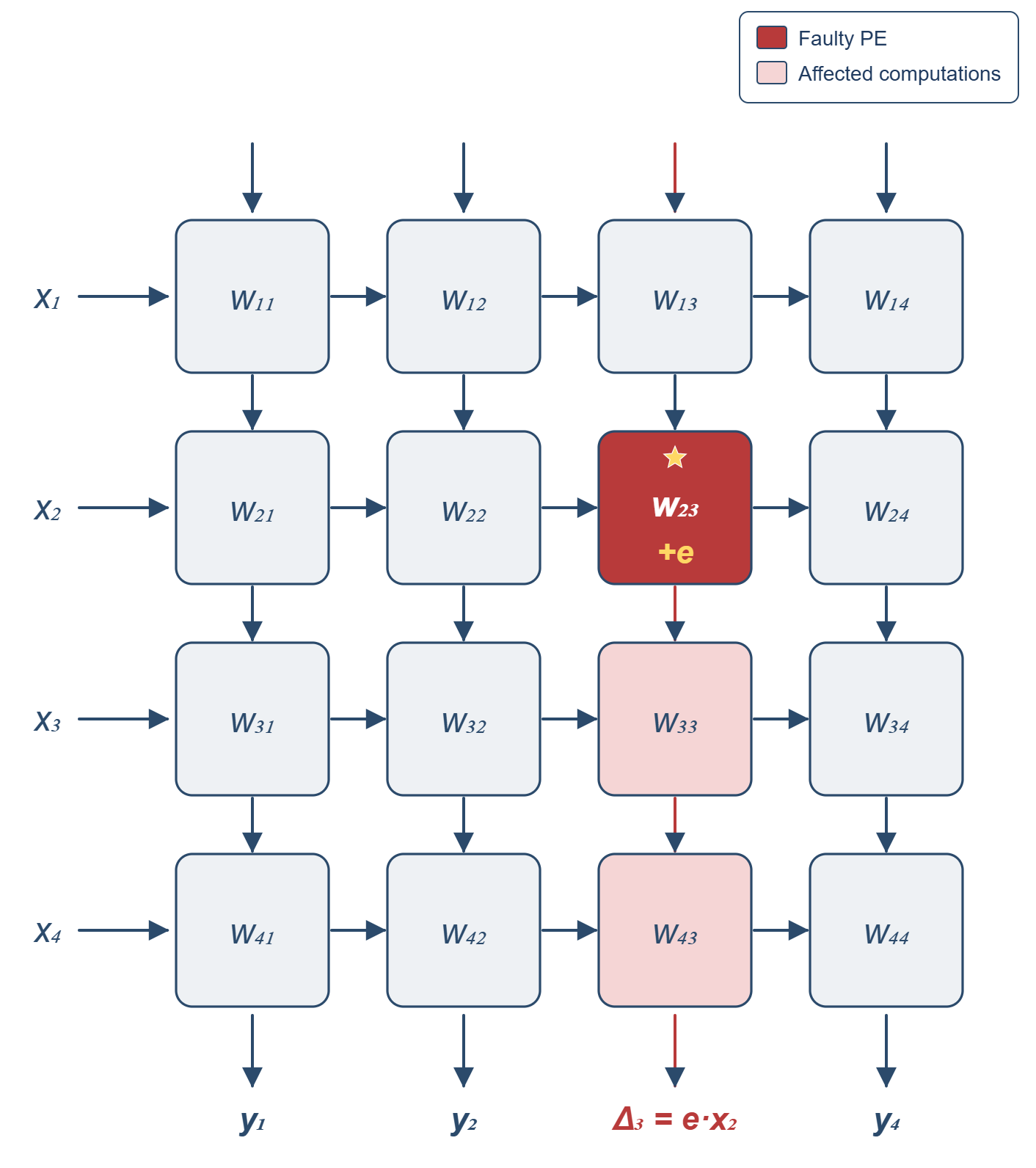}
    \caption{Error propagation pattern for a \textsf{Wreg} fault in a $4{\times}4$ weight-stationary systolic array, where a fault adds a persistent error $e$ that corrupts every MAC involving that PE.}
    \label{fig:wreg_fault}
    \vspace{-5mm}
\end{figure}

\paragraph{Bounded error model.}
For general faults (multi-bit defects, bridging faults, parametric drift, or silent data corruption) we make no assumption on the error value beyond a magnitude bound:
\begin{equation}\label{eq:bounded}
  e \;\sim\; \mathrm{Uniform}\bigl(\{1,\ldots,M\}\bigr), \qquad M = 2^{b-1}-1.
\end{equation}
This is the primary model under which we state all main results.

\paragraph{Single-bit-error model}
We also analyze our algorithms under a structured special case of the bounded error model: a permanent defect that forces bit~$\beta$ of the weight register to a fixed value $v\in\{0,1\}$, regardless of the data written, produces an error
\begin{equation}\label{eq:biterror}
 e = \pm\,2^{\beta},
\end{equation}
a signed power of two, whenever the intended weight has bit~$\beta$ equal to $1{-}v$; otherwise the fault is \emph{masked} ($e=0$).
Because single-bit errors have the restricted form $e=\pm 2^\beta$, our coprime scheme achieves strictly stronger guarantees than under the general bounded error model; we highlight these as remarks following the main theorems.

We state the main results under the bounded error model: Theorem~\ref{thm:avg} gives a one round probabilistic guarantee, and Section~\ref{sec:gcd} presents a two round ratio-based scheme for exact localization when one round is insufficient.
For the single-bit-error model, Remark~\ref{rem:singlebit} notes that exact localization in a single round is achievable.

\paragraph{Assumptions.}
\begin{enumerate}[label=\textbf{(A\arabic*)},leftmargin=*,nosep]
\item\label{A:single} At most one PE per column has a permanent fault.
\item\label{A:bounded} $1\le|e|\le M=2^{b-1}-1$.
\item\label{A:column} The faulty column is identified by $\Delta_{j^\ast}\neq 0$.
\end{enumerate}

Notably, we restrict attention to \emph{weight-register} faults, which result in the multiplicative structure \eqref{eq:deviation} that enables the algebraic localization techniques of our paper. 

\subsection{Problem Statement}
\label{sec:problem}

Consider a test phase in which a known weight matrix~$W$ has been loaded into the array and a fixed test input vector~$\mathbf{x}$ is applied.
The expected output $\mathbf{y} = W^T\mathbf{x}$ is precomputed offline.
If a permanent \textsf{Wreg} fault corrupts PE$(i^\ast,j^\ast)$, the observed output~$\hat{\mathbf{y}}$ differs from~$\mathbf{y}$ in column~$j^\ast$ by $\Delta_{j^\ast} = e\cdot x_{i^\ast}$ (Eq.~\ref{eq:deviation}), while all other columns match exactly.
\emph{Detecting} that a fault has occurred is therefore straightforward: any $\hat{y}_j \neq y_j$ signals a faulty column.

The harder problem and the focus of this paper is \emph{localization}: identifying \emph{which row}~$i^\ast$ within the faulty column contains the defective PE, and recovering the error value~$e$.

Since each column computes an independent inner product~\eqref{eq:output} and a \textsf{Wreg} fault affects only column~$j^\ast$, the localization problem decomposes across columns.
It therefore suffices to consider a single column with weight vector $\mathbf{w} = (w_1,\ldots,w_L)^T$, expected output $y = \mathbf{w}^T\mathbf{x}$, and observed output $\hat{y} = y + e\cdot x_{i^\ast}$.
We adopt this single-column notation for the remainder.
Formally:

\vspace{1mm}
\noindent
\textbf{Row Localization Problem.}
Design a small set of test input vectors $\mathbf{x}^{(1)},\ldots,\mathbf{x}^{(R)} \in \{-M,\ldots,M\}^L$, each with entries $\mathbf{x}^{(r)} = (x_1^{(r)},\ldots,x_L^{(r)})^T$, such that from the observed deviations $\Delta^{(r)} = e\cdot x_{i^\ast}^{(r)}$ ($r=1,\ldots,R$), the faulty row~$i^\ast$ and error~$e$ can be recovered exactly or with miss probability bounded by Theorem~\ref{thm:avg}. Table~\ref{tab:notation} lists notation.

\begin{table}[t]
\centering\small
\caption{Notation.\vspace{-10pt}}\label{tab:notation}
\begin{tabular}{@{}c@{\quad}p{0.62\columnwidth}@{}}
\toprule
$L,\,K$ & rows, columns of systolic array \\
$b$ & operand bit-width \\
$M$ & $2^{b-1}-1$; max signed magnitude \\
$e$ & additive error from \textsf{Wreg} fault \\
$\Delta_j$ & output deviation of column~$j$ \\
$\mathbf{x}=(x_1,\ldots,x_L)^T$ & test vector; entries pairwise coprime \\
$\pi(M)$ & Prime Counting Function: The number of prime numbers smaller than or equal to $M$ \\
$P_{\mathrm{fail}}$ & probability that localization is incomplete after one round \\
\bottomrule
\end{tabular}
\vspace{-5mm}
\end{table}

\section{Coprime Test-Vector Fault Localization}
\label{sec:method}

The scheme developed in this section is organized around a central idea: encoding fault location into arithmetic structure so that a small number of test passes recovers the faulty row from the observed deviation. We begin with a worked example (Section~\ref{sec:example}) and formalize the general framework (Section~\ref{sec:framework}). Section~\ref{sec:single} presents our main result: a \emph{single-round} construction based on pairwise coprime divisibility signatures, with a precise bound on the probability of successful localization (Theorem~\ref{thm:avg}). Section~\ref{sec:gcd} then develops a complementary two round construction that achieves exact localization. The two round construction serves two purposes: it provides exact localization as a fallback when one round is insufficient, and it offers a reference point against which the single-round result can be contrasted. Section~\ref{sec:overhead} quantifies the test overhead for both schemes.

\subsection{Motivating Example}
\label{sec:example}

Before stating the general framework, we trace through a small instance that illustrates every step.

\begin{example}[Signed INT4, $L=4$]\label{ex:int4}
Consider a single column with $L=4$ rows, signed $b=4$-bit operands ($M=2^{b-1}-1=7$), and weight vector $\mathbf{w}=(w_1,w_2,w_3,w_4)^T$.
Assign pairwise coprime test-vector entries, here the four primes $\le 7$:
$\mathbf{x} = (7,\;5,\;3,\;2)^T$.
The expected output is
$y = 7w_1 + 5w_2 + 3w_3 + 2w_4$.

\smallskip\noindent\textbf{Symbolic deviation.}\;
If row~$i^\ast=2$ has a \textsf{Wreg} fault with error~$e$, the observed output is
\[
  \hat{y} = 7w_1 + 5(w_2+e) + 3w_3 + 2w_4 = y + 5e,
\]
giving deviation $\Delta = \hat{y} - y = 5e = e\cdot p_2$.

\smallskip\noindent\textbf{Concrete instance.}\;
Load the all-ones test weight pattern $\mathbf{w}=(7,7,7,7)^T$, so $y = 7(7{+}5{+}3{+}2) = 119$.
A single-bit error on bit~2 of $w_2$ (bit forced to~$0$) changes $w_2 = 7 = 0\mathtt{b}0111$ to $0\mathtt{b}0011 = 3$, giving $e = 3-7 = -4$.
The observed output is $\hat{y} = 119 + 5\cdot(-4) = 99$, so $\Delta = 99 - 119 = -20$.

\smallskip\noindent\textbf{Localization.}\;
For each row~$k$, test whether $x_k\mid\Delta$ and check the implied error magnitude{, which must satisfy $|e|\le M = 7$ since operands are limited to $b=4$-bit precision}:
\begin{center}\small
\begin{tabular}{ccccc}
\toprule
Row $k$ & $x_k$ & $x_k\mid{-}20$? & $\Delta/x_k$ & $|\Delta/x_k|\le 7$? \\
\midrule
1 & 7 & No  & ---   & --- \\
2 & 5 & Yes & $-4$  & Yes $\checkmark$ \\
3 & 3 & No  & ---   & --- \\
4 & 2 & Yes & $-10$ & No \\
\bottomrule
\end{tabular}
\end{center}
Row~2 is the unique candidate: the scheme correctly localizes the fault and recovers $e=-4$.
Row~4 passes the divisibility test ($2\mid 20$) but is eliminated by the magnitude check ($|-10|>7$).
\end{example}

\noindent
This example illustrates that localization reduces to identifying which test vector entry divides the residual while respecting magnitude constraint. Two features are worth noting:
\vspace{1mm}
\begin{enumerate}[leftmargin=*,nosep]
\item \textbf{Why pairwise coprime?}  The deviation $\Delta = e\cdot c_{i^\ast}$ is a product of two factors.  Because the entries are pairwise coprime ($\gcd(c_k, c_{i^\ast})=1$ for $k\neq i^\ast$), a ``wrong'' entry $c_k$ can divide~$\Delta$ only if it divides the \emph{error}~$e$ itself. That is, row~$k$ survives the divisibility test only if $c_k$ divides $e$. Depending on the error model, we can bound the likelihood of the error $e$ being divisible by the wrong entry, which leads to the likelihood of localization with one test vector.
\item \textbf{Power-of-two errors.}  Single-bit errors produce errors $e=\pm 2^\beta$, which are powers of two.  No \emph{odd} integer coprime to~$2$ divides a power of two, so if we restrict to odd coprime entries, a single-bit error is localized \emph{with certainty} in a single round.  We formalize this below.
\end{enumerate}

\subsection{General Framework}
\label{sec:framework}

The test vector $\mathbf{x} = (x_1,\ldots,x_L)^T$ has entries that are \emph{pairwise coprime}: $\gcd(x_i,x_j)=1$ for all $i\neq j$.
Recall from Section~\ref{sec:problem} that we work with a single column: a fault in row~$i^\ast$ with error~$e$ produces deviation $\Delta = e\,x_{i^\ast}$.
The localization procedure checks, for each candidate row~$k$, whether $x_k \mid \Delta$.
If row~$k \neq i^\ast$ also passes this test, localization is incomplete: multiple candidate rows survive.

\begin{proposition}[Coprime divisibility test]\label{prop:div}
Let $\mathbf{x}$ have pairwise coprime entries $x_1,\ldots,x_L$.
Suppose the fault is in row~$i^\ast$ with error~$e \neq 0$.
Then row $k \neq i^\ast$ survives the divisibility test if and only if $x_k \mid e$.
In particular, if $e$ is not divisible by $x_k$ for any $k \neq i^\ast$, then $i^\ast$ is the unique survivor and localization is complete.
\end{proposition}
\begin{proof}
Row~$k$ survives the divisibility test iff $x_k \mid \Delta = e\,x_{i^\ast}$.
Since $\gcd(x_k,\,x_{i^\ast})=1$, this holds if and only if $x_k\mid e$.
\end{proof}

\noindent
Algorithm~\ref{alg:single} additionally checks the magnitude condition $|\Delta/x_k| \le M$, which can eliminate candidates that pass the divisibility test but imply an error exceeding the representable range (as seen for row~4 in Example~\ref{ex:int4}).%
\footnote{The analysis throughout uses only the divisibility condition $x_k \mid e$ and ignores the magnitude check. Since the magnitude check can only eliminate candidates, the stated bounds are conservative; a refined analysis incorporating this condition can only improve them.}

A key design implication of Proposition~\ref{prop:div}: since row~$k$ survives the divisibility test only if $x_k \mid e$, and the number of multiples of~$x_k$ in $\{1,\ldots,M\}$ is $\lfloor M/x_k \rfloor$, \emph{larger entries reduce the number of surviving candidates}.
For example, an entry $x_k = 113$ has only $\lfloor 127/113\rfloor = 1$ multiple in $\{1,\ldots,127\}$, while $x_k = 2$ has $63$ multiples.
This motivates choosing entries as large as possible.

An important constraint on the test vector length follows from pairwise coprimality.
If $x_1,\ldots,x_L$ are pairwise coprime integers with $x_k\le M$ for all~$k$, then each $x_k$ must contain at least one prime factor not shared with any other entry, and all such prime factors are $\le M$.
Therefore $L \le \pi(M)$, where $\pi(\cdot)$ denotes the prime-counting function (the number of primes up to its argument). 
The constraint $L \leq \pi(M)$ highlights a fundamental tradeoff: fault localization for larger arrays requires either additional probes or larger input magnitudes.

\subsection{One-Round Localization ($L \le \pi(M)$)}
\label{sec:single}

When $L \le \pi(M)$, a single test vector with $L$ pairwise coprime entries suffices.
To maximize $\min_k x_k$, we choose the $L$ largest primes $\le M$, though any set of $L$ pairwise coprime integers $\le M$ would work.
Algorithm~\ref{alg:single} gives the procedure.

\begin{algorithm}[t]
\caption{Single-Vector Coprime Localization}\label{alg:single}
\small
\KwIn{$L,K,M$; weight matrix~$W$}
\KwOut{candidate set $\mathcal{C}_j$ for each faulty column~$j$}
\BlankLine
Set $x_k \leftarrow$ the $k$-th largest prime $\le M$, for $k=1,\ldots,L$\;
$\mathbf{x}\leftarrow(x_1,\ldots,x_L)^T$\;
\ForEach{column $j\in[K]$}{
  $y_j^{\mathrm{exp}} \leftarrow \sum_{i} w_{ij}\,x_i$\;
}
Stream $\mathbf{x}$ through the systolic array; observe $y_j^{\mathrm{obs}}$\;
\ForEach{column $j$}{
  $\Delta_j\leftarrow y_j^{\mathrm{obs}}-y_j^{\mathrm{exp}}$\;
  \If{$\Delta_j\neq 0$}{
    $\mathcal{C}_j\leftarrow\{(k,\,\Delta_j/x_k): x_k\mid\Delta_j,\;|\Delta_j/x_k|\le M\}$\;
  }
}
\end{algorithm}

We now state the localization guarantees.

\begin{theorem}[Probability that localization is incomplete after one round]\label{thm:avg}
Let $L \le \pi(M)$ and assign pairwise coprime entries $x_1, \ldots, x_L \le M$.
Under the bounded error model with $e$ uniform on $\{1,\ldots,M\}$, the average probability that localization is incomplete after one round satisfies
\begin{equation}\label{eq:avgfp}
  P_{\mathrm{fail}}
  \;\le\; \frac{L-1}{LM}\sum_{k=1}^{L}\bigl\lfloor M/x_k\bigr\rfloor.
\end{equation}
When all entries exceed $M/2$, each $\lfloor M/x_k\rfloor = 1$ and this simplifies to $P_{\mathrm{fail}} \le (L{-}1)/M$.
\end{theorem}
\begin{proof}
By Proposition~\ref{prop:div}, row $k \neq i^\ast$ survives the divisibility test iff $x_k \mid e$.
The number of $e \in \{1,\ldots,M\}$ divisible by $x_k$ is $\lfloor M/x_k \rfloor$, so $\Pr(x_k \mid e) = \lfloor M/x_k \rfloor / M$.
By the union bound, $P_{\mathrm{fail}}(i^\ast) \le \frac{1}{M}\sum_{k \neq i^\ast} \lfloor M/x_k \rfloor$.
Averaging over $i^\ast$: each term $\lfloor M/x_k \rfloor$ appears in exactly $L{-}1$ of the $L$ inner sums (excluded only when $i^\ast = k$), giving $P_{\mathrm{fail}} \le \frac{L-1}{LM}\sum_{k=1}^{L}\lfloor M/x_k\rfloor$.
\end{proof}

Theorem~\ref{thm:avg} has two practical implications. First, it reinforces the design principle from Proposition~\ref{prop:div}: each term $\lfloor M/x_k\rfloor$ shrinks as $x_k$ grows, so choosing the largest available pairwise coprime entries minimizes the bound. Second, the bound improves sharply with operand bit-width. As $b$ (equivalently $M$) grows, the prime pool $\pi(M)$ expands, and for any fixed array size $L$ one can eventually select $L$ pairwise coprime entries all exceeding $M/2$. In this regime the bound collapses to the simplified form $P_{\mathrm{fail}} \le (L-1)/M$, which vanishes as $M \to \infty$. We explore this in Section~\ref{sec:inst} where we instantiate our results for INT8 and INT16.

\begin{remark}[Single-bit errors: exact single-round localization]\label{rem:singlebit}
Under the single-bit-error model ($e = \pm 2^\beta$), if all entries $x_k$ are chosen as the $L$ largest odd primes $\le M$ (requiring $L \le \pi(M)-1$), then $x_k \nmid e$ for every $k \neq i^\ast$, since $|e|$ is a power of two and every $x_k \ge 3$ is odd.
By Proposition~\ref{prop:div}, no other row survives the divisibility test and the faulty row is identified exactly with a single test vector.
\end{remark}

\subsection{Ratio-Based Two-Round Localization}
\label{sec:gcd}

The one round scheme of Section~\ref{sec:single} leaves two cases unresolved. First, for $L \le \pi(M)$, round~1 localizes with miss probability bounded by Theorem~\ref{thm:avg} but not with certainty: a fraction of faults produce errors $e$ that are not localized in one round. Second, for $L > \pi(M)$, a single test vector cannot have all entries pairwise coprime, and one round localization is not directly applicable. We now present a two round scheme that resolves both cases. It serves as a \emph{fallback} for round-1 failures when $L \le \pi(M)$ and as the \emph{primary} localization method when $L > \pi(M)$.

Given two test vectors $\mathbf{x}^{(1)}$ and $\mathbf{x}^{(2)}$ with row entries $x_k^{(1)}$ and $x_k^{(2)}$, a fault in row $i^\ast$ with error $e$ produces deviations
\begin{equation}
  \Delta^{(1)} = e \cdot x_{i^\ast}^{(1)},\qquad
  \Delta^{(2)} = e \cdot x_{i^\ast}^{(2)}.
\end{equation}
Taking the ratio cancels the error magnitude entirely:
\begin{equation}\label{eq:ratio}
  \frac{\Delta^{(2)}}{\Delta^{(1)}} \;=\; \frac{x_{i^\ast}^{(2)}}{x_{i^\ast}^{(1)}}.
\end{equation}
The key idea of the two round localization algorithm is to design test vectors so that every row has a distinct reduced ratio $r_k = x_k^{(2)}/x_k^{(1)}$, the observed ratio identifies the faulty row unambiguously. The error is then recovered by a single integer division, $e = \Delta^{(1)}/x_{i^\ast}^{(1)}$.

\paragraph{Feasibility constraint.}
Identification by ratio requires the per-row reduced ratios $r_k$ to be pairwise distinct. Since both $x_k^{(1)}$ and $x_k^{(2)}$ live in $\{1,\ldots,M\}$, the number of representable distinct ratios is finite. The set of admissible reduced ratios is
\begin{equation}\label{eq:farey}
  \mathcal{F}_M \;=\; \{p/q : 1 \le p,q \le M,\;\gcd(p,q)=1\},
\end{equation}
the set of all distinct reduced fractions with numerator and denominator in $\{1,\ldots,M\}$.%
\footnote{This set is closely related to the Farey sequence of order~$M$. Its cardinality equals the number of coprime pairs $(p,q)$ with $1 \le p,q \le M$, which is $2\sum_{q=1}^{M}\phi(q) - 1$ where $\phi$ is the Euler totient function. A classical result in number theory~\cite{hardywright2008} gives $\sum_{q=1}^{M}\phi(q) \sim 3M^2/\pi^2$, yielding $|\mathcal{F}_M| \approx 6M^2/\pi^2$.}
The ratio scheme therefore supports
\begin{equation}\label{eq:fareybound}
  L \;\le\; |\mathcal{F}_M| \;\approx\; \frac{6}{\pi^2}\,M^2.
\end{equation}
For INT8 this gives about 9800; for INT16, on the order of $6.5\times 10^8$. This is a theoretical upper bound; the two instantiations presented below do not saturate it but cover all array sizes of practical interest.

\noindent
The principle that two independent observations suffice to determine one error's location and value is classical, appearing in ABFT checksum schemes~\cite{huang1984abft} and in syndrome-based error correction~\cite{roth2006coding}. Its application to PE-level fault localization in systolic arrays, along with the implied bounds \eqref{eq:fareybound}, appears to be new. We present two useful instantiations: one for the case where it is used as a fallback for small array sizes to complement the one round test, and another for larger array sizes.

\paragraph{Instantiation 1: Derangement of primes (fallback for $L \le \pi(M)$).}Let $\mathbf{x}^{(1)} = (x_1, x_2, \ldots, x_L)$ be the round-1 test vector, whose entries are pairwise coprime. Let $\sigma$ be a derangement of $\{1,\ldots,L\}$ (a permutation with no fixed points, $\sigma(k) \ne k$), and set $\mathbf{x}^{(2)} = (x_{\sigma(1)}, x_{\sigma(2)}, \ldots, x_{\sigma(L)})$. Derangements exist for all $L \ge 2$, so this construction covers any $L \in [2, \pi(M)]$. The fact that the $x_1, x_2,\ldots, x_k$ are pairwise coprime implies that the ratios $\frac{x_k}{x_{\sigma(k)}}$ are all distinct. Critically, $\mathbf{x}^{(1)}$ is exactly the round-1 vector, so this instantiation is the natural fallback: if round 1 already localizes, the second pass is skipped; otherwise the second pass with $\mathbf{x}^{(2)}$ resolves the residual ambiguity via ratio recovery.

\paragraph{Instantiation 2: Consecutive integers (primary for $L > \pi(M)$).}
Let $\mathbf{x}^{(1)} = (1, 2, 3, \ldots, L)$ and $\mathbf{x}^{(2)} = (L, 1, 2, \ldots, L-1)$. The row pairs are $(1, L)$ and $(k, k-1)$ for $k = 2, \ldots, L$. Consecutive integers are coprime ($\gcd(k, k-1) = 1$ and $\gcd(1,L)=1$), and the $L$ ratios $L, 1/2, 2/3, \ldots, (L-1)/L$ are pairwise distinct. The construction requires $L \le M$ and uses no primes at all. This covers all practical array sizes at INT8 and above: $L \le 127$ at INT8 and $L \le 32767$ at INT16. For the slice between $M$ and the bound $|\mathcal{F}_M|$ in \eqref{eq:fareybound}, any enumeration of reduced fractions in $\mathcal{F}_M$ suffices.


\begin{example}[Ratio-based two round localization, signed INT4, $L=7$]\label{ex:gcd}
Consider $L=7$ rows with $b=4$ ($M=7$). Since $\pi(7) = 4 < 7$, one round is insufficient. 
\vspace{1mm}

\noindent Applying Instantiation 2: 
\vspace{1mm}

$\mathbf{x}^{(1)} = (1,2,3,4,5,6,7)^T$ and $\mathbf{x}^{(2)} = (7,1,2,3,4,5,6)^T$.

\vspace{1mm}
\noindent Row 1 receives the pair $(1,7)$; rows $k = 2, \ldots, 7$ receive $(k, k-1)$. The ratios are $7, 1/2, 2/3, 3/4, 4/5, 5/6, 6/7$ --- all distinct. Suppose row $i^\ast = 3$ has a fault with error $e = -2$.

\vspace{1mm}
\noindent\textbf{Round 1.}\;
$\Delta^{(1)} = e \cdot x_3^{(1)} = (-2) \cdot 3 = -6$.

\vspace{1mm}
\noindent\textbf{Round 2.}\;
$\Delta^{(2)} = e \cdot x_3^{(2)} = (-2) \cdot 2 = -4$.

\vspace{1mm}
\noindent\textbf{Ratio recovery.}\;
$\Delta^{(2)}/\Delta^{(1)} = (-4)/(-6) = 2/3$, which matches row $3$ uniquely (pair $(3,2)$). The error is recovered as $e = \Delta^{(1)}/x_3^{(1)} = -6/3 = -2$.
\end{example}

\subsection{Integration and Overhead}
\label{sec:overhead}

The localization scheme requires no new PE hardware or dedicated test datapaths.
A test pass is a standard matrix-vector multiply through the existing array:
the current weight tile stays resident in the PEs, the coprime sketch vector is
presented as a normal activation column, and the output is read back through the
standard output interface.
The golden checksum $\mathbf{c} = W_{\mathrm{gold}}\,\mathbf{x}$ is a
small per-tile precomputation stored alongside the weight tile in memory and
fetched on demand.
Post-processing runs on the array controller: compute $\boldsymbol{\Delta} =
\mathbf{s} - \mathbf{c}$, scan for the nonzero entry to find the faulty column $j^\ast$, then run $O(L)$ divisibility checks to locate the faulty row $i^\ast$.
This is $O(L\log M)$ bit operations in total: the $O(L)$ divisibility checks each cost $O(\log M)$ bit operations since every operand has $O(\log M)$ bits, keeping the total negligible relative to the array passes.

Deployment follows the same pattern as RunSAFER~\cite{vacca2023runsafer}:
the test is defined as an ISA instruction that queues the sketch vector(s)
into the activation pipeline alongside normal inference traffic.
Whether the trigger originates from an ISA opcode, a runtime scheduler, or a
host driver call is a deployment choice; the array-level operations are
identical in all cases. The ISA path is the natural production choice because it gives the controller deterministic timing.  Since the sketch
  vector is $1 \times L$, it matches the pipeline width exactly and is absorbed in one cycle with no stalls. As shown in Section~\ref{sec:cycleoverhead}, the incremental cost of issuing a sketch vector through an already-loaded array is exactly 1 cycle regardless of array size, making the overhead inherently negligible.

\section{Evaluation}
\label{sec:inst}

We evaluate FLARE along two axes: analytical characterization of
localization probability, and empirical validation via fault-injection
simulation.
For the analytical part we instantiate the $P_{\mathrm{fail}}$ bounds of
Theorem~\ref{thm:avg} for INT8 and INT16, comparing raw-prime and
prime-power coprime pool constructions.
For the empirical part we validate these bounds on a custom cycle-accurate
systolic array simulator and measure cycle overhead with
SCALE-Sim~3.0.0~\cite{raj2025scale}, reporting the incremental cost over a
baseline of one inference tile without the sketch vector.
Array dimensions are chosen to include the $128{\times}128$ and
$256{\times}256$ configurations of the Google TPU~\cite{jouppi2017tpu}
as reference deployment targets.

\subsection{Analytical Results}

We instantiate the $P_{\mathrm{fail}}$ bounds of Theorem~\ref{thm:avg} for INT16 and INT8, comparing the raw-prime and prime-power coprime pool constructions for each.
\subsubsection{INT16 ($b=16$, $M=32767$)}
\label{sec:int16}

At INT16 the prime pool is $\pi(32767) = 3512$, and $1612$ of these primes exceed $M/2$. Assigning the $L$ largest primes to the rows, the simplified bound $P_{\mathrm{fail}} \le (L-1)/M$ from Theorem~\ref{thm:avg} therefore applies for all $L \le 1612$, which covers the $128{\times}128$ and $256{\times}256$ array dimensions of the Google TPU~\cite{jouppi2017tpu}. For $L=256$ this evaluates to under $0.008$, and for $L=1024$ to under $0.032$. The single-round scheme (Algorithm~\ref{alg:single}) suffices at INT16 with no need for the prime-power construction or two round fallback developed for INT8 below. Under the single-bit-error model (Remark~\ref{rem:singlebit}), restricting to odd primes yields exact single-round localization.

\begin{figure*}[t]
    \centering
    \begin{subfigure}[t]{0.48\textwidth}
        \centering
        \includegraphics[width=\linewidth]{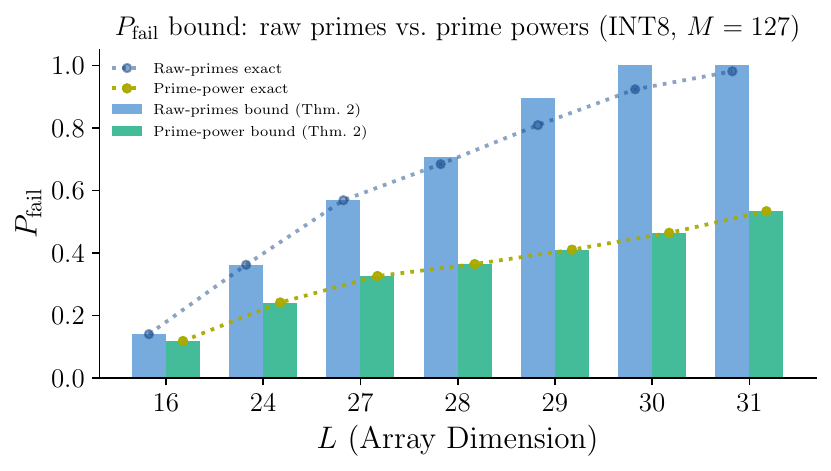}
        \caption{Average $P_{\mathrm{fail}}$ bound (Theorem~\ref{thm:avg}) for raw primes and prime powers. The prime-power exact value coincides with its bound, confirming no overcounting. The raw-prime exact value falls below its bound, showing the union bound overcounts when small primes are present.}
        \label{fig:primes_vs_powers_top}
        \vspace{-5mm}
    \end{subfigure}
    \hfill
    \begin{subfigure}[t]{0.48\textwidth}
        \centering
        \includegraphics[width=\linewidth]{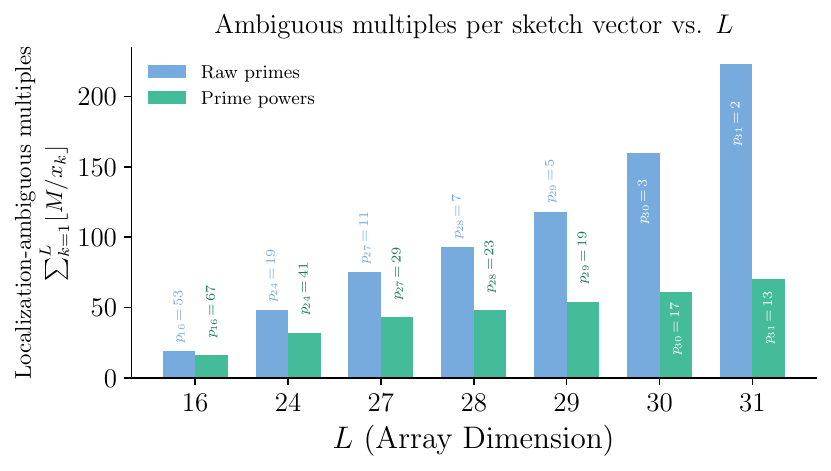}
        \caption{Ambiguous multiple count in the sketch vector across array dimensions $L$; $p_L$ is the additional entry added to the test vector at each $L$. The prime-power pool yields far fewer ambiguous multiples than raw primes, reducing false localization candidates and $P_{\mathrm{fail}}$ as seen above.}
        \label{fig:primes_vs_powers_bottom}
        \vspace{-5mm}
    \end{subfigure}
    \vspace{2mm}
    \caption{Raw primes vs.\ prime-power coprimes for INT8 ($M = 127$). Beyond $L=27$, small primes enter the raw-prime pool, raising the ambiguous multiple count; more error values become indistinguishable from the true fault, introducing false localization candidates and increasing $P_{\mathrm{fail}}$; the union bound also overcounts for raw primes, with the exact value falling below it.}
    \label{fig:primes_vs_powers}
    \vspace{-3mm}
\end{figure*}

\subsubsection{INT8 ($b=8$, $M=127$)}
\label{sec:int8}

The 31~primes $\le 127$ are
\vspace{0.5mm}
\[
\begin{gathered}  2,3,5,7,11,13,17,19,23,29,31,37,41,43,47,\\
  53,59,61,67,71,73,79,83,89,97,101,103,107,109,113,127.
\end{gathered}
\]

\paragraph{Bounded error model.}
Using all 31~primes, Theorem~\ref{thm:avg} gives the average $P_{\mathrm{fail}}$ via~\eqref{eq:avgfp}.
The 13~largest primes all exceed $M/2=63$; each contributes $\lfloor M/x_k\rfloor = 1$.
Evaluating~\eqref{eq:avgfp}:
\begin{center}\small
\begin{tabular}{lcc}
\toprule
$L$ & primes used & $P_{\mathrm{fail}}$ (Theorem~\ref{thm:avg})\\
\midrule
$8$ & $\{89,\ldots,127\}$ & {$\le 7/127 \approx 0.055$} \\
$16$ & $\{53,\ldots,127\}$ & {$\le 15/127 \approx 0.118$} \\
$31$ & all (prime-power) & {$\le 2100/3937 \approx 0.534$} \\
\bottomrule
\end{tabular}
\end{center}
\paragraph{Sharper construction: prime powers as coprime entries.}
One round construction of Section \ref{sec:single} and the corresponding analysis of Theorem~\ref{thm:avg} requires only that the test entries be \emph{pairwise coprime}, not that they be prime. This freedom can be exploited: each small prime $p$ can be replaced by the largest power $p^k \le M$ that fits in the error domain, without violating coprimality (since each power involves only one prime base). For $M=127$ the replacements are
\[
\begin{gathered}
2 \to 2^6=64,\; 3 \to 3^4=81,\; 5 \to 5^3=125,\;\\ 7 \to 7^2=49,\; 11 \to 11^2=121
\end{gathered}
\]
while primes $\ge 13$ remain unchanged. The resulting pool still has $\pi(127)=31$ entries, but the smallest entry rises from 2 to 13. The number of multiples in $\{1,\ldots,127\}$ collapses dramatically: where the original prime $2$ contributed 63 multiples, $64$ contributes only $\lfloor 127/64\rfloor = 1$. Similarly $81$, $125$, $121$ each contribute only one multiple, and $49$ contributes two. Substituting into Theorem~\ref{thm:avg} reduces $P_{\mathrm{fail}}$ for $L = 31$ from impractical (the small primes dominate the union bound) to $2100/3937 \approx 0.534$. Figure~\ref{fig:primes_vs_powers} shows the empirical impact: for a single sketch vector applied to INT8 ($M=127$), the prime power pool produces a noticeably lower average $P_{\mathrm{fail}}$ than the raw prime pool, with the gap widening beyond $L=27$ where the raw prime pool is forced to include the small primes (11, 7, 5, 3, 2) that dominate the union bound.

The union bound of Theorem~\ref{thm:avg} overcounts when $e$ is divisible by multiple pool entries, as with the raw-prime pool where entries like 2 and 3 share common multiples in $\{1,\ldots,M\}$; the exact $P_{\mathrm{fail}}$ (Appendix~\ref{app:exact}) corrects for this and falls strictly below the bound. The prime-power pool eliminates this overcounting: since no two entries share a common multiple in $\{1,\ldots,M\}$, the bound coincides with the exact $P_{\mathrm{fail}}$ computation for prime-power pool entries (Figures~\ref{fig:primes_vs_powers_top} and~\ref{fig:pfail_comp_combined}). Beyond tightness, the union bound evaluates in $O(L \log M)$ bit operations (Section~\ref{sec:overhead}), whereas the exact $P_{\mathrm{fail}}$ via inclusion-exclusion (Appendix~\ref{app:exact}) scales exponentially in $L$ and becomes impractical for large arrays. We therefore use the Theorem~\ref{thm:avg} bound with the prime-power pool throughout, as it matches the exact $P_{\mathrm{fail}}$ and remains a tight estimate across all tested dimensions.

\begin{figure*}[t]
    \centering
    \begin{subfigure}[t]{0.48\textwidth}
        \centering
        \includegraphics[width=\linewidth]{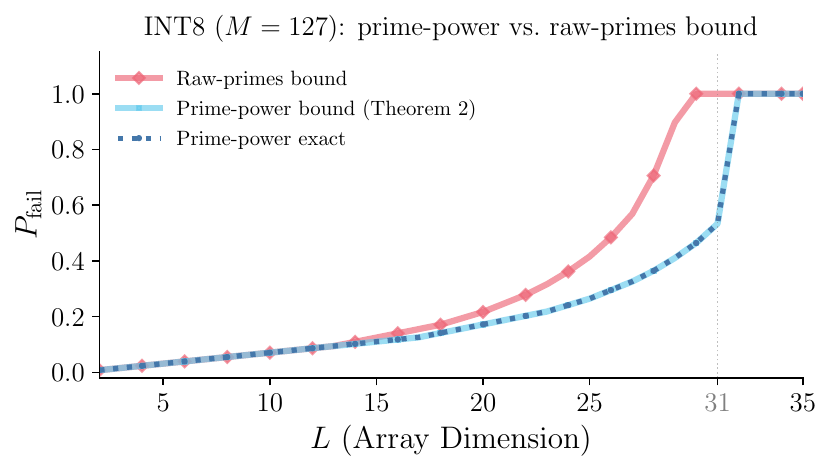}
        \caption{INT8  ($M = 127$, $\pi(127) = 31$): gap is most pronounced near $L = 31$, where the raw-prime pool is forced to include small primes with many multiples in $\{1,\ldots,M\}$.}
        \label{fig:pfail_comp_int8}
        \vspace{-5mm}
    \end{subfigure}
    \hfill
    \begin{subfigure}[t]{0.48\textwidth}
        \centering
        \includegraphics[width=\linewidth]{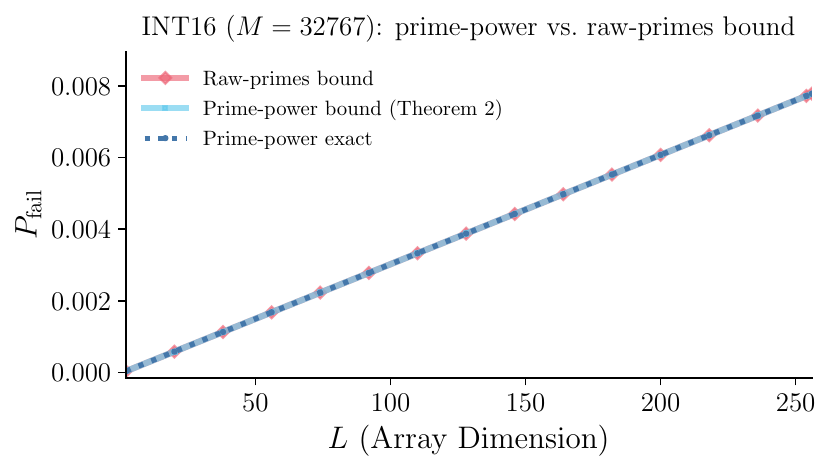}
        \caption{INT16 ($M = 32767$, $\pi(32767) = 3512$): all curves remain negligible for $L \le 256$, confirming single-round sufficiency.}
        \label{fig:pfail_comp_int16}
        \vspace{-5mm}
    \end{subfigure}
    \vspace{2mm}
    \caption{Prime-power exact (blue dotted), prime-power bound (cyan, Theorem~\ref{thm:avg}), and raw-primes bound (rose) for INT8 and INT16. The prime-power exact and bound coincide in both cases, confirming no overcounting.}
    \label{fig:pfail_comp_combined}
    \vspace{-3mm}
\end{figure*}

\paragraph{Production-scale arrays ($L\in\{128,256\}$).}
For $L > 31 = \pi(127)$, the one round scheme cannot assign a unique coprime entry to every row. The ratio-based two round scheme (Section~\ref{sec:gcd}) provides exact localization for almost all practical INT8 array sizes, since $|\mathcal{F}_{127}| \approx 9{,}800$ (Eq.~\ref{eq:fareybound}) far exceeds most production array dimensions.

\paragraph{Single-bit-error model.}
By Remark~\ref{rem:singlebit}, using the 30~odd primes ($3$ through~$127$) guarantees exact localization in one round for any error $e=\pm 2^\beta$.

\vspace{-2mm}
\paragraph{Cross-precision comparison.}
Precision has a large effect on localization reliability: the pool size difference ($\pi(127)=31$ vs.\ $\pi(32767)=3512$) means that, at any given array dimension, INT16 achieves substantially lower $P_{\mathrm{fail}}$ than INT8. Concretely, INT8 reaches $P_{\mathrm{fail}} \approx 0.534$ at its pool boundary ($L=31$), while INT16 remains below $0.008$ even at $L=256$, as shown in Figures~\ref{fig:pfail_comp_int8} and~\ref{fig:pfail_comp_int16}.

\subsection{Empirical Results}
\label{sec:experiments}

Experiments are run on a custom Python cycle-accurate simulator using weights from the \path{model.layers.0.mlp.down\_proj} tensor of Qwen2.5-0.5B~\cite{yang2025qwen3}. We also measure $P_{\mathrm{fail}}$ empirically across array dimensions, and characterize the cycle cost of the sketch vector pass for both the one-round and two-round schemes using SCALE-Sim~3.0.0~\cite{raj2025scale}.

\subsubsection{Empirical Validation of $P_{\mathrm{fail}}$}

For each array dimension $L$ and precision, a fresh $L$-row weight-stationary systolic array is instantiated and loaded with the first $L$ rows of \path{model.layers.0.mlp.down\_proj} of Qwen2.5-0.5B~\cite{yang2025qwen3} (full tensor shape $896{\times}4864$): for INT8, weights are taken from the W8A8 quantized version \path{neuralmagic/Qwen2.5-0.5B-Instruct-quantized.w8a8}, stored as signed 8-bit integers; for INT16, from the BF16 version \path{Qwen/Qwen2.5-0.5B-Instruct}, quantized to signed 16-bit integers.\footnote{The weights do not enter the $P_{\mathrm{fail}}$ computation: localization is determined entirely by the syndrome $x_{i^\ast} \cdot e$. Therefore, only the precision of the random error $e$ and the sketch vector affect $P_{\mathrm{fail}}$, through the magnitude range $\{1,\ldots,M\}$ they induce, and not the precision or values of the weight matrix. This was validated using the custom cycle-accurate Python simulator.} For each of the 500 trials, the fault row $i^\ast\!\in\!\{0,\dots,L{-}1\}$, fault column $j^\ast\!\in\!\{0,\dots,K{-}1\}$, and error magnitude $e\!\in\!\{1,\dots,M\}$ are sampled uniformly at random; the sketch vector is constructed from the prime-power pool for the corresponding precision (Section~\ref{sec:single}). An additive fault of size $e$ is injected into the \textsf{Wreg} of PE $(i^\ast,j^\ast)$, the prime-power coprime sketch vector is passed through the faulty array, and the candidate row set is recovered from the output syndrome. 

Localization is incomplete when, in addition to $i^\ast$, at least one other row appears in the candidate set of Algorithm~\ref{alg:single}. The empirical failure rate is compared against the Theorem~\ref{thm:avg} prime-power bound. We sweep $L\in\{4,8,16,31,40,50,60\}$ for INT8 and $L\in\{4,8,16,32,64,128,256\}$ for INT16.

Figures~\ref{fig:empirical_pfail_int8} and~\ref{fig:empirical_pfail_int16} report results for the INT8 ($M=127$) and INT16 ($M=32767$) error domains respectively. The shaded band around each simulation curve shows the 95\% confidence interval~\cite{wilson1927} for the 500-trial sample. For INT8, the empirical rate tracks the Theorem~\ref{thm:avg} bound closely across all tested dimensions; for $L>31=\pi(127)$ the prime-power pool is exhausted and both the Theorem~\ref{thm:avg} bound and the exact $P_{\mathrm{fail}}$ saturate to~1 analytically, though the empirical failure rate degrades more gradually (discussed below), motivating the two-round scheme of Section~\ref{sec:gcd}. Occasional crossings where the empirical estimate slightly exceeds the bound are expected: the bound is a guarantee on the true average $P_{\mathrm{fail}}$, while the empirical estimate from a finite number of trials is subject to sampling noise, so small upward fluctuations do not contradict the theoretical guarantee. For INT16 the failure rate is negligible across all practical array sizes, confirming that the prime-power construction provides reliable single-round localization up to $256{\times}256$ arrays at higher precision.

\paragraph{Graceful degradation beyond the pool boundary.}
A closer look at the INT8 curve reveals that the empirical $P_{\mathrm{fail}}$ remains noticeably below~1 for $L$ slightly above~31. Both the Theorem~\ref{thm:avg} union bound and the exact $P_{\mathrm{fail}}$ (Appendix~\ref{app:exact}) saturate to~1 as soon as $L > \pi(127) = 31$ because entry/prime-factor repetition is unavoidable, but repetition does not affect every row equally: for $L = 31 + k$, exactly $k$ rows must share a prime-power entry with another row, while the remaining~$31-k$ rows still carry unique entries. Since the fault location $i^\ast$ is drawn uniformly over all $L$ rows, the probability that it falls on a non-repeated row is {$(31-k)/L$}, and single-round localization succeeds in those cases. The empirical $P_{\mathrm{fail}}$ therefore rises by at most {$2k/L$} due to entry repetition, rather than jumping immediately to~1. {Concretely, for $L = 32$ ($k = 1$), exactly $2/32$ of fault locations fall on shared-entry rows and fail with certainty; the remaining $30/32$ still benefit from single-round localization.} This graceful degradation is practically useful: INT8 arrays moderately larger than~31 rows can still benefit from single-round FLARE localization for the majority of fault locations, and the two-round fallback is only strictly necessary once enough rows carry repeated entries to meaningfully raise the overall failure rate.

\begin{figure*}[t]
    \centering
    \begin{subfigure}[t]{0.48\textwidth}
        \centering
        \includegraphics[width=\linewidth]{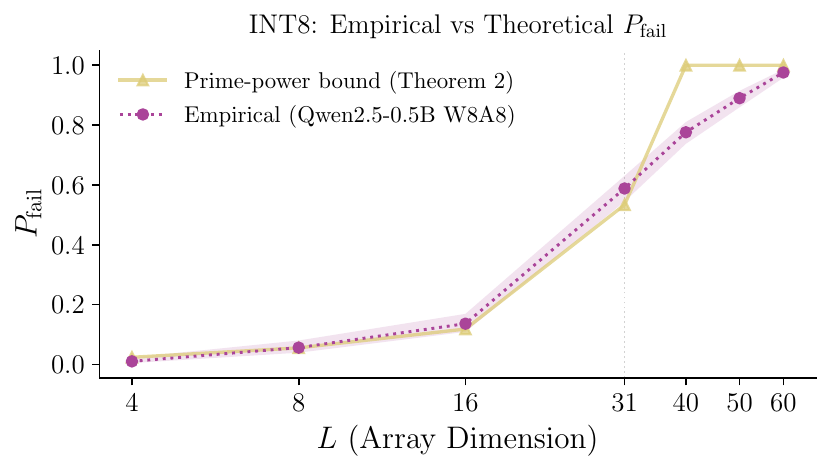}
        \caption{INT8 ($M = 127$): For $L > 31 = \pi(127)$ the prime-power pool is exhausted and the analytical $P_{\mathrm{fail}}$ saturates to~1, though the empirical rate rises gradually as non-repeated entries continue to localize successfully for sizes moderately beyond the boundary.}
        \label{fig:empirical_pfail_int8}
        \vspace{-5mm}
    \end{subfigure}
    \hfill
    \begin{subfigure}[t]{0.48\textwidth}
        \centering
        \includegraphics[width=\linewidth]{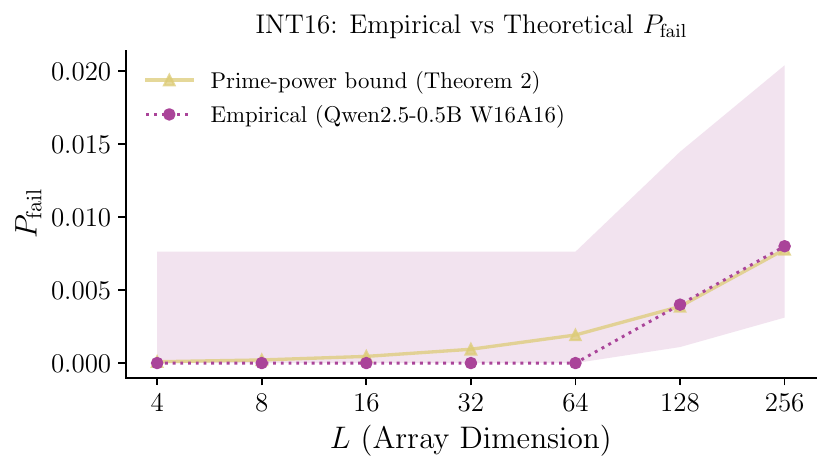}
        \caption{INT16 ($M = 32767$): The failure rate remains negligible across all tested dimensions, consistent with the Theorem~\ref{thm:avg} bound.}
        \label{fig:empirical_pfail_int16}
        \vspace{-5mm}
    \end{subfigure}
    \vspace{2mm}
    \caption{Empirical $P_{\mathrm{fail}}$ versus array dimension $L$ for INT8 and INT16 with Theorem~\ref{thm:avg} prime-power bound (500 trials; Qwen2.5-0.5B weights: W8A8 for INT8, BF16 quantized to INT16 for INT16); shaded regions denote 95\% confidence bands~\cite{wilson1927}.}
    \label{fig:empirical_pfail_combined}
    \vspace{-3mm}
\end{figure*}

\vspace{-1mm}
\subsubsection{Cycle Overhead}
\label{sec:cycleoverhead}

We measure the cycle cost of issuing the sketch vector(s) relative to a standard inference tile, for both the one-round and two-round schemes across array dimensions.
In the ISA-integrated deployment model (Section~\ref{sec:overhead}), the weights
  for the tile under test are already resident in the PEs when the sketch vector
  is issued.
  The relevant overhead is therefore the \emph{incremental} cost of issuing the
  sketch vector through the already-loaded array.
  The sketch vector is $1 \times L$, one element per PE column of the
  $L \times L$ array, making a sketch pass a GEMM of shape
  $(L \times L) \times (L \times 1)$.

  We measure this using SCALE-Sim~3.0.0~\cite{raj2025scale}.
  The Baseline column of Table~\ref{tab:cycle_overhead} reports the cycle count for one inference tile (including weight loading) as measured by SCALE-Sim; this is subtracted from the total cycle count with the sketch vector to isolate the test overhead alone.
  For the two round scheme both sketch vectors are issued back-to-back against
  the same baseline.

  Table~\ref{tab:cycle_overhead} reports the results.
  The sketch vector has the same length as the array column count and the pipeline absorbs each additional vector
  in exactly \textbf{1 cycle}, regardless of array size.
  The two round scheme therefore costs 2 cycles total.
  At $128{\times}128$ and $256{\times}256$ the single-round overhead is $0.04\%$ and $0.01\%$ respectively; the two-round overhead is $0.08\%$ and $0.03\%$.
  The post-pipeline detection work (syndrome subtraction and $O(L)$ divisibility checks) runs asynchronously on the adjacent array controller and can overlap with subsequent tile computation. The fault-identification latency is $O(L)$ controller clock cycles; whether this overlaps fully with array execution depends on the controller's arithmetic throughput, which is implementation-specific and not characterized here. The reported cycle counts therefore reflect only the array-side cost of issuing the sketch vector(s).

  \begin{table}[t]
    \centering
    \caption{Incremental cycle overhead of FLARE on an $L \times L$
             Weight-Stationary array, measured with SCALE-Sim~3.0.0~\cite{raj2025scale}.
             The sketch vector is $1 \times L$;
             Util.\,(\%) is the PE utilization of the baseline run (fraction of PE-cycles
             actively performing a MAC, as reported by SCALE-Sim);
             Overhead\,\% $= \Delta\text{cycles} / \text{baseline} \times 100$. \vspace{-5pt}}
    \label{tab:cycle_overhead}
    \small
    \setlength{\tabcolsep}{5pt}
    \begin{tabular}{rrrrrrr}
      \toprule
      $L$ & Baseline & Util.\,(\%) &
      \multicolumn{2}{c}{1-round} &
      \multicolumn{2}{c}{2-round} \\
      \cmidrule(lr){4-5}\cmidrule(lr){6-7}
      & & & cyc. & \% & cyc. & \% \\
      \midrule
        8 &  554 & 3.45 & 1 & 0.18 & 2 & 0.36 \\
       16 &  586 & 1.64 & 1 & 0.17 & 2 & 0.34 \\
       32 &  650 & 0.80 & 1 & 0.15 & 2 & 0.31 \\
       64 &  778 & 0.40 & 1 & 0.13 & 2 & 0.26 \\
      128 & 2606 & 0.20 & 1 & 0.04 & 2 & 0.08 \\
      256 & 7836 & 0.10 & 1 & 0.01 & 2 & 0.03 \\
      \bottomrule
    \end{tabular}
    \vspace{-5mm}
  \end{table}
\vspace{-5pt}

\section{Conclusion}
\label{sec:conclusion}
We presented FLARE, a lightweight algorithmic framework for PE-level fault localization in systolic arrays without hardware redundancy. By leveraging pairwise coprime test vectors, FLARE encodes fault location into algebraic divisibility signatures that survive systolic accumulation, overcoming the fundamental limitation of prior approaches that only achieve column-level detection. We show that a single test pass localizes faults with miss probability bounded by Theorem~\ref{thm:avg} under a general bounded error model, while a two round ratio-based extension guarantees exact localization for all practical array sizes. Importantly, FLARE requires no architectural modifications and integrates seamlessly with existing execution pipelines. Experimental results demonstrate near-perfect localization for INT16 and strong guarantees for INT8, with test overhead under 1\% of a single GEMM tile, making it practical for in-field deployment. Overall, FLARE highlights the potential of algorithm-hardware co-design for low-cost, software-defined fault diagnosis in AI accelerators.

\begin{acks}
We acknowledge NSF awards \#2506573 and \#2516418 for this work. We acknowledge Prof.\ Tushar Krishna at Georgia Tech ECE for helpful discussions. We also thank Ajay Sharma Mandadi at Georgia Tech for valuable feedback throughout this work.
\end{acks}

\bibliographystyle{unsrt}
\bibliography{references}

\appendix

\section{Exact Probability via Inclusion-Exclusion}
\label{app:exact}

The bound of Theorem~\ref{thm:avg} applies the union bound, summing $\lfloor M/x_k \rfloor$ over all rows $k \neq i^\ast$. This can overcount when $e$ is simultaneously divisible by more than one pool entry -- for example, $e=6$ is divisible by both $x=2$ and $x=3$, thus counted twice, even though it corresponds to a single failure event. The exact average $P_{\mathrm{fail}}$ corrects for this via the inclusion-exclusion principle.

\medskip
\noindent\textbf{Setup.}\; Localization is incomplete when at least one row $k \neq i^\ast$ survives the divisibility test. The syndrome at the faulty column is $\Delta = e \cdot x_{i^\ast}$, so row $k$ passes the test when $x_k \mid e \cdot x_{i^\ast}$; since pool entries are pairwise coprime ($\gcd(x_k, x_{i^\ast}) = 1$ for $k \neq i^\ast$), this simplifies to $x_k \mid e$. Since $e$ is bounded by the precision range $\{1,\ldots,M\}$, the failure-inducing values are precisely the multiples of $x_k$ within this range: the errors that render row $k$ indistinguishable from $i^\ast$.

\smallskip
\noindent Let $A_k = \{e \in \{1,\ldots,M\} : x_k \mid e\}$ be the set of such errors for row $k$. Taking the union over all rows $k \neq i^\ast$ collects all errors that would cause at least one wrong row to survive alongside the true faulty row. For a fixed faulty row $i^\ast$:

\smallskip
\begin{equation}
  P_{\mathrm{fail}}(i^\ast) = \frac{1}{M}\left|\bigcup_{k \neq i^\ast} A_k\right|.
\end{equation}

\medskip
\noindent\textbf{Inclusion-exclusion formula.}\; Let $S$ denote a non-empty subset of non-faulty row indices drawn from $\{1,\ldots,L\}\setminus\{i^\ast\}$; each such subset represents a combination of rows that would simultaneously pass the divisibility test, and its contribution counts the errors $e \le M$ divisible by all entries $x_k$ for $k \in S$, i.e., $\left|\bigcap_{k\in S} A_k\right| = \lfloor M/\prod_{k\in S} x_k \rfloor$. Summing with alternating signs over all such subsets gives the exact union cardinality:
\begin{equation}\label{eq:exact}
  \left|\bigcup_{k \neq i^\ast} A_k\right|
  = \sum_{j=1}^{L-1}(-1)^{j+1}
    \sum_{\substack{S \subseteq \{1,\ldots,L\}\setminus\{i^\ast\} \\ |S|=j}}
    \left\lfloor \frac{M}{\prod_{k \in S} x_k} \right\rfloor.
\end{equation}
Averaging~\eqref{eq:exact} uniformly over $i^\ast \in \{1,\ldots,L\}$ gives the average exact $P_{\mathrm{fail}}$, plotted as \emph{exact} in Figures~\ref{fig:primes_vs_powers} and~\ref{fig:pfail_comp_combined}.

\medskip
\noindent\textbf{Prime-power pool.}\; For the prime-power pool, all pairwise products $x_j x_k$ exceed $M$, so no $e \le M$ is simultaneously divisible by two distinct entries and overcounting cannot occur. Consequently, formula~\eqref{eq:exact} collapses to the union bound of Theorem~\ref{thm:avg}: the bound and exact value coincide, justifying its use throughout the paper.

\medskip
\noindent\textbf{Computational complexity.}\; The summation in~\eqref{eq:exact} ranges over all non-empty subsets of $\{1,\ldots,L\}\setminus\{i^\ast\}$, where every subset is formed by choosing elements in all possible combinations. Since the faulty row $i^\ast$ is excluded from subset formation, there are $L-1$ non-faulty rows to choose from. The number of subsets of size $j$ from these $L-1$ rows is $\binom{L-1}{j} = \frac{(L-1)!}{j!\,(L-1-j)!}$, and the total number of non-empty subsets across all sizes is $\sum_{j=1}^{L-1}\binom{L-1}{j} = 2^{L-1}-1$ by the binomial theorem, which grows exponentially in $L$.

\medskip
\noindent For each subset $S$ of size $j$, computing the product $\prod_{k\in S}x_k$ costs $O(j\log M)$ bit operations and the subsequent floor division costs $O(\log M)$. Averaging over all $L$ choices of $i^\ast$ multiplies the total cost by $L$, giving $O(L \cdot 2^{L} \cdot \log M)$ bit operations overall. In contrast, the Theorem~\ref{thm:avg} union bound requires only $O(L)$ floor divisions totalling $O(L\log M)$, independent of the subset structure. For large $L$ the exact formula is therefore computationally prohibitive, further supporting the use of the union bound in practice.

\end{document}